\documentclass[final,3p,times]{elsarticle}
\usepackage[colorlinks=true]{hyperref}
\usepackage{amsthm,amsmath,amssymb,mathrsfs,amsfonts,graphicx,graphics,latexsym,exscale,cmmib57,dsfont,bbm,amscd,ulem}
\usepackage{color,soul}%%%\hl{ highlighting text } % use hl{ } to highlight the sentence.
\newcommand{\bA}{\boldsymbol{A}}
\newcommand{\bL}{\boldsymbol{L}}

\newcommand{\bF}{\boldsymbol{F}}

\newcommand{\balpha}{\boldsymbol{\alpha}}

\newcommand{\bcdot}{\boldsymbol{\cdot}}
\newcommand{\blambda}{\boldsymbol{\lambda}}
\newcommand{\bchi}{\boldsymbol{\chi}}
\newcommand{\bdelta}{\boldsymbol{\delta}}
\newcommand{\bmu}{\boldsymbol{\mu}}

\newcommand{\bgamma}{\boldsymbol{\gamma}}
\newtheorem{thm}{Theorem}[section]
\theoremstyle{definition}
\newtheorem{defn}[thm]{Definition}
\newtheorem{Example}[thm]{Example}

\newcommand{\Real}{\mathbb R}

\numberwithin{equation}{section}
\begin{document}

\begin{frontmatter}

\title{Stability of time-varying nonlinear switching systems under perturbations\tnoteref{t1}}
\tnotetext[t1]{This work was supported in part by National Science Foundation of China (Grant Nos. 11071112 and 11071263), PAPD of Jiangsu Higher Education and in part by NSF 1021203 of the United States.}

\author[Dai]{Xiongping Dai}
\ead{xpdai@nju.edu.cn}
\author[Huang]{Yu Huang}
\ead{stshyu@mail.sysu.edu.cn}
\author[LX]{Mingqing Xiao}
\ead{mxiao@math.siu.edu}

%%\cortext[cor1]{Corresponding author}

\address[Dai]{Department of Mathematics, Nanjing University, Nanjing 210093, People's Republic of China}
\address[Huang]{Department of Mathematics, Zhongshan (Sun Yat-Sen) University, Guangzhou 510275, People's Republic of China}
\address[LX]{Department of Mathematics, Southern Illinois University, Carbondale, IL 62901-4408, USA}
%%\address[LX]{Department of Mathematics, Southern Illinois University, Carbondale, IL 62901-4408, USA}

%%%%%%%%%
\begin{abstract}
Recently, switched systems have been found in many practical applications such as in communication network dynamics, robot manipulators, traffic management, etc.
In this paper, we study the stability of time-varying linear switched systems under perturbations that satisfy linear growth condition, which is not available in current literature.
By introducing a new Liao-type exponent for time-varying linear switched systems, we generalize our recent result given in \cite{DHX}, and provide a sufficient condition of exponential stability for quasi-linear switched systems. The criterion of asymptotic stability is computable and is applicable to switched systems that may consist
of infinitely many time-dependent subsystems.
\end{abstract}

\end{frontmatter}

%%%%%%%%%%%%%%%%%%%%%%%%%%%%%%%%%%%%%%%%%%%%%%%%%%%%%%%%%%%%%%%%%%%%%%%%%%%
%%%%%%%%%%%%%%%%%%%%%%%%%%%%%%%%%%%%%%%%%%%%%%%%%%%%%%%%%%%%%%%%%%%%%%%%%%%
\section{Introduction}\label{sec1}

Let $X,\mathcal{I}$ be two topological spaces and we denote by $\mathrm{C}^0(X,\mathbb{R}^n)$ the
set of all continuous functions from $X $ into $\mathbb{R}^n$ endowed with the uniform-convergence topology; i.e., the topology is induced by the metric \begin{equation*}
|f-g|=\sup_{x\in X}\|f(x)-g(x)\|\quad\forall f,g\in\mathrm{C}^0(X,\mathbb{R}^n).
\end{equation*}
Here $\|\cdot\|$ denotes the usual Euclidean norm on $\mathbb{R}^n$. Without any confusion, we also use $\|\cdot\|$ for the matrix norm on $\mathbb{R}^{n\times n}$ induced by the corresponding Euclidean norm. Throughout this paper we denote $\mathbb{R}_+=(0,+\infty)$ as the $t$-time space.
Let $d\ge 2$ be an integer, then two continuous function-valued functions are defined as
\begin{equation*}
\bA\colon\mathcal{I}\rightarrow\mathrm{C}^0(\mathbb{R}_+,\mathbb{R}^{d\times d});\;
i\mapsto A_i(\bcdot)=\left[A_i^{\ell m}(\bcdot)\right]\quad \textrm{and}\quad
\bF\colon\mathcal{I}\rightarrow
\mathrm{C}^0(\mathbb{R}_+\times\mathbb{R}^d,\mathbb{R}^d);\; i\mapsto
F_i(\bcdot,\bcdot)
\end{equation*}
where $\bA, \bF$ satisfy
\begin{equation*}
\|A_i(t)\|\le\balpha\quad \textrm{and}\quad \|F_i(t,x)\|\le
\bL\|x\|\quad\forall (t,x)\in\mathbb{R}_+\times\mathbb{R}^d
\end{equation*}
for some $\balpha>0$ and $\bL>0$, uniformly for $i\in\mathcal {I}$. We notice here that
$F_i(t,x)$ is not necessarily linear with respect to the space variable $x\in\mathbb{R}^d$.

Let $\sigma_{(\bcdot)}\colon\mathbb{N}\rightarrow\mathcal{I}$ and $\tau_{(\bcdot)}\colon\mathbb{N}\rightarrow\mathbb{R}_+$ be two infinite sequence, where $\tau$ is  strictly increasing with $\tau_n\uparrow+\infty$ and $\mathbb{N}=\{1,2,\dotsc\}$. Then each pair $(\sigma,\tau)$ gives rise to a piecewise constant, left continuous switching law $u_{\sigma,\tau}$ with the switching-time sequence $\tau$ as follows:
\begin{equation*}
u_{\sigma,\tau}\colon\mathbb{R}_+\rightarrow\mathcal {I};\quad
u_{\sigma,\tau}(t)=\sigma(n)\quad\textrm{ for }\tau_{n-1}<t\le \tau_n\textrm{ and }
n\in\mathbb{N}.\qquad \textrm{Here }\tau_0:=0.
\end{equation*}
A quasi-linear switched system associated with a switching control can be expressed in the form of
\begin{equation}\label{eq1.1}
\dot{x}(t)=A_{u_{\sigma,\tau}(t)}(t)x(t)+F_{u_{\sigma,\tau}(t)}(t,x(t)),\quad t\in\mathbb{R}_+\textrm{ and }
x(0)=x_0\in\mathbb{R}^d.
\end{equation}
In this paper, we study the stabilization problem of above system by using the switching control $u_{\sigma,\tau}(t)$ and develop a sufficient condition for $u_{\sigma,\tau}(t)$ which can warrant
the global exponential stability of the system.

The interest of this goal has been primarily motivated due to this type of model is governed by many man-made and natural systems in mathematics,
control engineering, biology and physics. Studies of this type of problems fall under the category of
``stabilization analysis" in the control theory. To the best of our knowledge, the study of time-varying switched systems with infinitely many subsystems is not available in current literature.

In order to study the stability of (\ref{eq1.1}) with a class of $F_i(\bcdot,\bcdot)$, an effective approach in literature is the perturbation method. That is, instead of considering (\ref{eq1.1}) directly, we consider the asymptotic stability of its linear approximation system
\begin{equation}\label{eq1.2}
\dot{v}(t)=A_{u_{\sigma,\tau}(t)}(t)v(t),\qquad
v(0)=v_0\in\mathbb{R}^d\textrm{ and }
t\in\mathbb{R}_+
\end{equation}
If a control $u_{\sigma,\tau}(t)$ can stabilize the above system, then we study (\ref{eq1.1}) by viewing  $F_{u_{\sigma,\tau}(t)}(t,x)$ as a perturbation of (\ref{eq1.2}) to seek under what condition the asymptotic stability can still be maintained.

Given a switching control $u_{\sigma,\tau}(t)$, the fundamental characteristic for
the stability of the switching system (\ref{eq1.2})
is the maximal Lyapunov exponent, whose mathematical expression is given by
\begin{equation}\label{eq1.3}
\blambda(u_{\sigma,\tau})=\max_{v_0\in\mathbb{R}^d}\left\{\limsup_{t\to+\infty}\frac{1}{t}\log\|v(t,v_0, u_{\sigma,\tau})\|\right\}<0.
\end{equation}
Here $v(\bcdot, v_0,u_{\sigma,\tau})\colon\mathbb{R}_+\rightarrow\mathbb{R}^d$, such that $v(0, v_0,u_{\sigma,\tau})=v_0$,
denotes the solution of (\ref{eq1.2}).
However, it is well known in classical literature that for time-varying systems, its Lyapunov exponent in general is not robust under (even arbitrary small) perturbation.
Thus it is not suitable for us to use
the Lyapunov exponent of (\ref{eq1.2}) to ``approximate'' (\ref{eq1.1}) to determine its asymptotic stability.
Even when $A_i(t)=A_i$ is a constant matrix for each $i\in\mathcal{I}$, the system (\ref{eq1.1}) becomes
\begin{equation}\label{eq1.4}
\dot{x}(t)=A_{u_{\sigma,\tau}(t)}x(t)+F_{u_{\sigma,\tau}(t)}(t,x(t)),\quad t\in\mathbb{R}_+\textrm{ and }
x(0)=x_0\in\mathbb{R}^d,
\end{equation}
whose dynamics behavior is known to be similar to the time-varying systems due to the switching action. Thus the same issue mentioned above still remains.
In light of this, in \cite{DHX}
the authors introduced a so-called {\it Liao-type exponent} for (\ref{eq1.4})
when $A_i, \forall i\in\mathcal {I}$, is upper triangular, and show that this Liao-type exponent is robust with respect to small perturbation.

Since we here are interested in the time-varying switched systems (\ref{eq1.1}), the Liao-type exponent defined in \cite{DHX} is no longer to be valid. This motivates us to look for a new type of exponent that can carry the asymptotic stability from the linear approximation (\ref{eq1.2}) to (\ref{eq1.1}) with the appearance of perturbation $\bF$. In this paper, we will develop a new Liao-type exponent for linear system (\ref{eq1.2}) which can provide three new important properties: (i) it is robust to many perturbations that satisfy linear growth condition; (ii) it can capture the stability even if subsystems have unstable modes; (iii) it includes the Liao-type exponent  defined in \cite{DHX} as a special case.

The importance of properties (i) and (iii) are readily to be seen according to the goal of this paper. Let us make some elaborations on property (ii).
For simplicity, we let the control set $\mathcal {I}=\{1,\dotsc,K\}$ be finite, $F_i(t,x)=F_i(x)=\textrm{o}(\|x\|)$ for $1\le i\le K$, and assume $\bmu$ is the
finite-dimensional stationary distribution of the associated Markovian
chain $\varSigma_K^+$, and set
\begin{equation*}
[i]_1=\{\sigma=(\sigma_n)_{n=1}^{+\infty}\in\varSigma_{K}^+\,|\, \sigma_1=i\}\quad \textrm{for }1\le i\le K.
\end{equation*}
Standard approach is to define a positive definite, norm-like
function $V$, called the Lyapunov function, on $\mathbb{R}^d$ such that
$V(x(t))$ is a decreasing function of $t$ for all solutions $x(t)$
of (\ref{eq1.4}). For this, one needs to
seek a symmetric and positive definite matrix $G$ such that the mean value of the maximum eigenvalues in terms of $\bmu$ is negative, i.e.,
\begin{equation}\label{eq1.6}
\sum_{1\le i\le K}\blambda_{\max}(GA_iG^{-1}+G^{-1}A_i^{\mathrm{T}}G)\bmu([i]_1)<0;
\end{equation}
see, e.g.,  \cite{KK,SWMWK,ZYS} for more details.  This condition requires that there exists at least one index $i\in\{1,\dotsc,K\}$ such that $$\blambda_{\max}(GA_iG^{-1}+G^{-1}A_i^{\mathrm{T}}G)<0,$$
which is equivalent to $\blambda_{\max}(A_i+A_i^{\mathrm{T}})<0$. Thus a necessary condition for this
approach is to require at least one
subsystem to be dissipative; in other words, there is a real positive constant $\gamma$ such that $\|e^{A_it}\|\le e^{-\gamma t}\,\forall t>0$ for some index $i\in\{1,\dotsc,K\}$. If this necessary condition is not satisfied, for example when all subsystems have unstable models, then the approach mentioned above cannot be applied.  To require at least one subsystem to be dissipative seems too restrictive since even an exponentially stable subsystem may not be dissipative. This is because $A_i$ is stable $\mathit{iff}$ $\|e^{A_it}\|\le C e^{-\gamma t}\;\forall t>0$ for some constants $C\ge1$ and $\gamma>0$. Here $C$ does not need to be equal to $1$.

The rest of the paper is arranged as follows. In Section~\ref{sec2}
we will provide a new, more general definition of the Liao-type exponents than that introduced in \cite{DHX}.
In Section~\ref{sec3} we will provide a computable criterion of asymptotic exponential stability, shown in Theorem~\ref{thm3.1}. Finally, the paper ends with concluding remarks in Section~\ref{sec4}.

%%%%%%%%%%%%%%%%%%%%%%%%%%%%%%%%%%%%%%%%%%%%%%%%%%%%%%%%%%%%%%%%%%%%%
%%%%%%%%%%%%%%%%%%%%%%%%%%%%%%%%%%%%%%%%%%%%%%%%%%%%%%%%%%%%%%%%%%%%%
\section{The Liao-type exponents}\label{sec2}%

In this section, we introduce the new Liao-type exponent that is more general than the one given by~\cite{DHX},
for switched systems
consisting of infinitely many non-autonomous
subsystems:
\begin{equation*}
\dot{x}=A_i(t)x+F_i(t,x),\quad(t,x)\in\mathbb{R}_+\times\mathbb{R}^d\textrm{ and }
i\in\mathcal{I},
\end{equation*}
where, for each control value $i\in\mathcal {I}$,
$A_i(t)=\left[A_i^{\ell m}(t)\right]\in\mathbb{R}^{d\times d}$ is a continuous upper triangular
matrix-valued function of $t$ and $F_i(t,x)\in\mathbb{R}^d$
is continuous with respect to $(t,x)$, such that
\begin{equation*}
\|A_i(t)x\|\le\balpha\|x\|\; \forall (t,x)\in\mathbb{R}_+\times\mathbb{R}^d\quad\textrm{and}\quad\|F_i(t,x)\|\le
\bL(t)\|x\|\;\forall x\in\mathbb{R}^d
\end{equation*}
where $\balpha, \bL(t)$ both are
independent of the indices $i\in\mathcal{I}$.

Given any $T_*>0$, a sequence of positive real numbers $\tau=\{\tau_n\}_{n=1}^{+\infty}$ is called a $T_*$-switching-time sequence, if it is a switching-time sequence (i.e., $0<\tau_1<\tau_2<\dotsm$ and $\tau_n\uparrow+\infty$) and such that
$\tau_n-\tau_{n-1}\le T_{\!*}$ for all $n\ge 1$.

In what follows, we let $u\colon\mathbb{R}_+\rightarrow\mathcal {I}$ be a
switching law associated with $\sigma\colon\mathbb{N}\rightarrow\mathcal{I}$ and a $T_*$-switching-time sequence $\tau=\{\tau_n\}_{n=1}^{+\infty}$.
Then, $u$ defines a quasi-linear switching system
\begin{equation*}
\dot{x}(t)=A_{u(t)}(t)x(t)+F_{u(t)}(t,x(t)),\qquad(t,x)\in\mathbb{R}_+\times\mathbb{R}^d.\eqno{(\bA,\bF)_{u}}
\end{equation*}

\begin{defn}\label{def2.1}%%%
Let $\{n_s\}_{s=1}^{+\infty}$ be an arbitrarily given integer sequence
such that
\begin{equation*}
1\le n_s-n_{s-1}\le\varDelta
\quad\forall s\in\mathbb{N},\quad(n_0:=0),
\end{equation*}
where $\varDelta$ is a positive integer. Associated to this sequence, for $(\bA,\bF)_{u}$
the real number
\begin{equation*}
\bchi_*^+(\bA_{u})=\limsup_{s\to+\infty}\frac{1}{\tau_{n_s}}\sum_{k=1}^s\max_{1\le
j\le d}\left\{\int_{\tau_{n_{k-1}}}^{\tau_{n_k}}A_{u(t)}^{jj}(t)\,dt\right\}
\end{equation*}
is called a {\it Liao-type exponent} of $\bA_{u}$. Here $\tau=\{\tau_n\}_{n=1}^{+\infty}$ is a $T_*$-switching sequence determined by $u$.
\end{defn}

Notice that the new defined Liao-type exponent $\bchi_*^+$ depends on (i) the switching control $u(t)$; (ii) the duration period $[\tau_{n_{s-1}}, \tau_{n_s}]$ on the corresponding subsystems. In particular, if we choose $\{n_s=s\}_{s=1}^{+\infty}$, then $\bchi_*^+$ is just equal to the Liao-type exponent $\bchi$ defined in \cite{DHX} when the switched system is time invariant. Thus, the new definition is a generalization of the
one introduced by Dai, Huang and Xiao in \cite{DHX}.

As the ``linear approximation" of $(\bA,\bF)_{u}$, we next consider the following linear switching system $\bA_{u}$:
\begin{equation*}
\dot{v}(t)=A_{u(t)}(t)v(t),\qquad t\in\mathbb{R}_+\textrm{ and }
v(0)=v_0\in\mathbb{R}^d.
\end{equation*}
Then,
\begin{equation*}
\blambda(\bA_{u}):=\max_{1\le j\le
d}\left\{\limsup_{T\to+\infty}\frac{1}{T}\int_0^T
A_{u(t)}^{jj}(t)\,dt\right\}
\end{equation*}
is the (maximal) Lyapunov exponent of the linear system
$\bA_{u}$.

From these definitions, we have
$\blambda(\bA_{u})\le\bchi_*^+(\bA_{u})$. On the other hand, the following example shows that the new Liao-type exponent captures the stability better
than the one defined in \cite{DHX}.
Let us see an example.

\begin{Example}\label{example2.2}%%%%
Let $\mathcal{I}=\{0,1\}$ and
\begin{equation*}
{A}_0=\left[\begin{matrix}1&0\\
0&-2\end{matrix}\right],\quad {A}_1=\left[\begin{matrix}-2&0\\
0&1\end{matrix}\right]
\end{equation*}
and define
$\sigma=(\uwave{0,1},\uwave{0,1},\dotsc)$
and $\tau=\{\tau_n=n\}_{n=1}^{+\infty}$. For the switching system $\bA_{u_{\sigma,\tau}}$:
\begin{equation*}
\dot{v}(t)={A}_{u_{\sigma,\tau}(t)}v(t),\qquad v(0)=v_0\in\mathbb{R}^2\textrm{ and }
t\in\mathbb{R}_+,
\end{equation*}
it is easy to see that although each subsystem has an unstable mode, the system $\bA_{u_{\sigma,\tau}}$, associated with this periodically switching control $u_{\sigma,\tau}(t)$, is exponentially stable.
By a simple calculation, we can get that $\blambda(\bA_{u_{\sigma,\tau}})=-1/2$, and
associated to the sequence $\{n_s=2s\}_{s=1}^{+\infty}$,
$\bchi_*^+(\bA_{u_{\sigma,\tau}})=-1/2$. However, the previously proposed Liao-type exponent in \cite{DHX} equals $1$.
\end{Example}

In addition, it is essential to see that, as is shown by
Example~\ref{example2.2}, $\bchi_*^+(\bA_{u})<0$ does not
need to imply any subsystems of $\bA_{u}$ stable.

%%%%%%%%%%%%%%%%%%%%%%%%%%%%%%%%%%%%%%%%%%%%%%%%%%%%%%%%%%%%%%%%%%
%%%%%%%%%%%%%%%%%%%%%%%%%%%%%%%%%%%%%%%%%%%%%%%%%%%%%%%%%%%%%%%%%%
\section{A criterion of global exponential stability}\label{sec3}%%%

In this section, we will present a criterion of asymptotic, exponential
stability for the switching system
$(\bA,\bF)_{u}$ introduced in Section~\ref{sec2}, which is an extension of \cite[Theorem~2.2]{DHX}.

\begin{thm}\label{thm3.1}%%%%
Let $A_i(t)\in\mathbb{R}^{d\times d}$ be upper-triangular for each
$i\in\mathcal {I}$. Assume $\bA_{u}$ has the Liao-type
exponent $\bchi_*^+(\bA_{u})<0$ associated to a $\varDelta$-sequence
$\{n_s\}_{s=1}^{+\infty}$. Then, there exists a constant $\bdelta>0$ such that whenever $\bL(t)\le
L<\bdelta$ for $t$ sufficiently large, the switching system
$(\bA,\bF)_u$ is globally, asymptotically, exponentially stable.

If $A_i(t)=\mathrm{diag}(A_i^{11}(t),\dotsc, A_i^{dd}(t))$ for all $i\in\mathcal{I}$, then $\bdelta$ can be defined by
\begin{equation*}
\bdelta=|\bchi_*^+(\bA_{u})|\exp(-2\bgamma\varDelta T_{\!*})
\end{equation*}
where
\begin{equation*}
\bgamma=\sup\left\{A_{u(t)}^{jj}(t)\,|\,t>0, 1\le j\le
d\right\}-\inf\left\{A_{u(t)}^{jj}(t)\,|\,t>0, 1\le j\le d\right\}.
\end{equation*}
\end{thm}

\begin{proof}
The following proof is motivated by the one given in~\cite{DHX}. The approach is
a subtle combination of the new Liao-type exponent and Lyapunov functions.

Notice that there is no loss of generality in assuming
that $\bL(t)\le L$ for any $t>0$. Next, we will
first show the case when
$A_i(t)=\mathrm{diag}\left(A_i^{11}(t),\dotsc,A_i^{dd}(t)\right)$
for all $i\in\mathcal {I}$. The upper-triangular case will be
discussed afterwards.

Let $\{\tau_n\}_{n=1}^{+\infty}$ be the $T_*$-switching time sequence of $u$ as in Definition~\ref{def2.1}. We define a sequence of constants $\{\chi_s^+\}_{s=1}^{+\infty}$ by
\begin{equation}\label{eq2.1}
\chi_s^+=\max_{1\le j\le
d}\left\{\frac{1}{\tau_{n_s}-\tau_{n_{s-1}}}\int_{\tau_{n_{s-1}}}^{\tau_{n_s}}
A_{u(t)}^{jj}(t)\,dt\right\}
\end{equation}
and define, for $1\le j\le d$ and $s=1,\dotsc$, continuous
functions
\begin{subequations}\label{eq2.2}
\begin{align}
h_s^{jj}(t)&=\exp\left\{(t-\tau_{n_{s-1}})\chi_s^+-\int_{\tau_{n_{s-1}}}^t
A_{u(\tau)}^{jj}(\tau)\,d\tau\right\}\quad \textrm{for }\tau_{n_{s-1}}<t\le \tau_{n_s},\label{eq2.2a}\\
\intertext{which are piecewise differentiable. Let}
H_s(t)&=\mathrm{diag}\left(h_s^{11}(t),\dotsc, h_s^{dd}(t)
\right)\qquad \textrm{for }\tau_{n_{s-1}}<t\le \tau_{n_s}\label{eq2.2b}
\end{align}
\end{subequations}
be the diagonal $d$-by-$d$ matrix for $s=1,\dotsc$. Then from
(\ref{eq2.2a}), it follows that
\begin{equation}\label{eq2.3}
\sup\limits_{s\in{\mathbb N}}\left\{\max\limits_{1\le j\le
d}\left\{\sup\limits_{\tau_{n_{s-1}}<t\le \tau_{n_s}}\left\{h_s^{jj}(t),
{h_s^{jj}(t)}^{-1}\right\}\right\}\right\}\le\exp(\bgamma\varDelta
T_{\!*}),
\end{equation}
where $\bgamma$ is defined as in the theorem and $\varDelta$ is an upper bound of $n_s-n_{s-1}$ described in Definition~\ref{def2.1}.
Notice that for each $s\ge1$,
\begin{equation}\label{eq2.4}
h_s^{jj}(\tau_{n_{s-1}})=1\quad \textrm{and}\quad
h_s^{jj}(\tau_{n_s})\ge1.
\end{equation}
By the non-autonomous linear transformations of variables
\begin{equation}\label{eq2.5}
y=H_s(t)x,\quad\tau_{n_{s-1}}<t\le \tau_{n_s}
\end{equation}
for $s=1,\dotsc$, now $(\bA,\bF)_{u}$ restricted to
$(\tau_{n_{s-1}},\tau_{n_s}]$ is transformed into the following
quasi-linear system
\begin{equation}\label{eq2.6}
\dot{y}(t)=\overline{A}_s(t)
y(t)+\overline{F}_s(t,y(t)),\quad \tau_{n_{s-1}}<t\le \tau_{n_s}\textrm{ and }
y\in\mathbb{R}^d,
\end{equation}
where for $\tau_{n_{s-1}}<t\le \tau_{n_s}$, $\overline{A}_s(t)$ and
$\overline{F}_s(t,y)$ satisfy, respectively,
\begin{subequations}\label{eq2.7}
\begin{align}
\overline{A}_s(t)&=A_{u(t)}(t)+\frac{d^-H_s(t)}{dt}{H_s(t)}^{-1}\label{eq2.7a}\\
\intertext{and}\overline{F}_s(t,y)&=H_s(t)F_{u(t)}\left(t,{H_s(t)}^{-1}y\right).\label{eq2.7b}
\end{align}
\end{subequations}
Notice here that $\frac{d^-}{dt}$ denotes $\frac{d}{dt}$ at any
regular time $t$ and the left-derivative at a switching time
$t=\tau_n$. According to (\ref{eq2.7a}) and (\ref{eq2.2a}), a direct
calculation yields
\begin{equation}\label{eq2.8}
\overline{A}_s(t)\equiv\mathrm{diag}\left(\chi_s^+,\dotsc,\chi_s^+
\right)\quad \textrm{for }\tau_{n_{s-1}}<t\le \tau_{n_s},
\end{equation}
for all $s\in\mathbb{N}$.

Next, we will prove that (\ref{eq2.6}) satisfies some important estimation.
For any $s\in\mathbb{N}$, any $\tau_{n_{s-1}}<t\le \tau_{n_s}$, and any
$y=(y_1,\dotsc,y_d)^\textrm{T}\in\Real^d$, (\ref{eq2.7b}) together
with (\ref{eq2.3}) leads to
\begin{equation}\label{eq2.9}
\|\overline{F}_s(t,y)\|\le\|F_{u(t)}\left(t,
H_s(t)^{-1}y\right)\|\exp(\bgamma\varDelta T_{\!*}).
\end{equation}
Accordingly, it is easily seen that for any nonzero $x=H_s(t)^{-1}y$,
we obtain by (\ref{eq2.3})
\begin{equation}\label{eq2.10}
\frac{\|\overline{F}_s(t,y)\|}{\|y\|}\le\frac{\|F_{u(t)}(t,
x)\|\exp(2\bgamma\varDelta T_{\!*})}{\|x\|}\le L\exp(2\bgamma\varDelta
T_{\!*})
\end{equation}
for any $\tau_{n_{s-1}}<t\le \tau_{n_s}$, for all $s\in\mathbb{N}$. Denote
\begin{equation}\label{eq2.11}
\overline{F}_s(t,y)=\left(\bar{f}_{s,1}(t,y),\dotsc,\bar{f}_{s,d}(t,y)\right)^\textrm{T}\quad\forall s\in\mathbb{N},\
y\in\mathbb{R}^d,\ t\in(\tau_{n_{s-1}}, \tau_{n_s}].
\end{equation}
Hence, from the Cauchy inequality, it follows that for all
$s\in\mathbb{N}$, $\tau_{n_{s-1}}<t\le \tau_{n_s}$, and any
$y=(y_1,\dotsc,y_d)^\textrm{T}\in\Real^d$
\begin{equation}\label{eq2.12}
\left|\sum_{j=1}^dy_j\bar{f}_{s,j}(t,y)\right|\le
\left\{\sum_{j=1}^dy_j^2\right\}^{\frac12}\left\{\sum_{j=1}^d\bar{f}_{s,j}^2(t,y)\right\}^{\frac12}
=L\exp(2\bgamma\varDelta T_{\!*})\|y\|^2.
\end{equation}
Now, take arbitrarily a constant $\varepsilon$ with
$0<\varepsilon<1$. Define
\begin{equation}\label{eq2.13}
{\bdelta}_\varepsilon=\varepsilon|\bchi_*^+(\bA_{u})|\exp(-2\bgamma\varDelta
T_{\!*}).
\end{equation}
Then from (\ref{eq2.12}), the condition
$L\le{\bdelta}_\varepsilon$ yields
\begin{equation}\label{eq2.14}
\left|\sum_{j=1}^n{y_j}\bar{f}_{s,j}(t,y)\right|\le\varepsilon|\bchi_*^+(\bA_{u})|\cdot\|y\|^2
\end{equation}
for any $s=1,2,\dotsc$, $\tau_{n_{s-1}}<t\le \tau_{n_s}$, and any
$y=(y_1,\dotsc,y_d)^\textrm{T}\in{\mathbb R}^d$.

Hereafter, let
$L\le{\bdelta}_\varepsilon$ and $x_0\in{\mathbb
R}^d$ be arbitrarily taken. Let
\begin{equation*}
x(t)=x(t,x_0)=\left(x_1(t,x_0),\dotsc,x_d(t,x_0)\right)^{\mathrm
T}\in\mathbb R^d
\end{equation*}
be a solution of $(\bA,\bF)_{u}$ with $x(0)=x_0$, which is
continuous and piecewise differentiable in $t\in\mathbb{R}_+$. If
$x_0=\mathbf{0}$ the origin of $\mathbb{R}^d$, then
$x(t)\equiv\mathbf{0}$ for all $t>0$ from a Osgood-type uniqueness
theorem \cite[Lemma~2.5]{DHX}.

Next, assume $x_0\not=\mathbf{0}$ and so $x(t)\not=\mathbf{0}$ for
all $t>0$. For all $s\in\mathbb{N}$, we now define the Lyapunov
functions as follows:
\begin{equation}\label{eq2.15}
V_s(t)=\frac12\sum_{j=1}^dy_{s,j}^2(t) \quad \textrm{for }\tau_{n_{s-1}}\le t\leq
\tau_{n_s},
\end{equation}
where for each $s\in\mathbb{N}$,
\begin{equation*}
y_s(t)=\left(y_{s,1}(t),\dotsc,y_{s,d}(t)\right)^\textrm{T}=H_s(t)x(t,x_0)\quad \textrm{for }\tau_{n_{s-1}}<t\le \tau_{n_s},
\end{equation*}
and at the time instant $t=\tau_{n_{s-1}}$
\begin{equation*}
V_s(\tau_{n_{s-1}})=\lim_{t\downarrow \tau_{n_{s-1}}}V_s(t).
\end{equation*}
Thus, by (\ref{eq2.4}) we have
\begin{equation}\label{eq2.16}
V_s(T_{n_{s}})\ge V_s(\tau_{n_{s-1}}),\quad V_s(t)>0\quad
\textrm{and}\quad
\dot{y}_s(t)=\overline{A}_s(t)y_s(t)+\overline{F}_s(t,y_s(t))
\end{equation}
for any $\tau_{n_{s-1}}<t\le \tau_{n_s}$ and any $s=1,2,\dotsc$. This
together with (\ref{eq2.8}) yields that
\begin{equation}\label{eq2.17}
\dot{y}_{s,j}(t)={\chi_s^+}y_{s,j}(t)+\bar{f}_{s,j}(t,y_s(t)),\quad
\tau_{n_{s-1}}<t\le \tau_{n_s},
\end{equation}
for each $j=1,\dotsc,d$, where $\bar{f}_{s,j}(t,y_s(t))$ is defined
in the same way as in (\ref{eq2.11}).

Then from (\ref{eq2.15}), (\ref{eq2.16}) and (\ref{eq2.14}), it
follows that for any $s=1,2,\dotsc$ and any $\tau_{n_{s-1}}<t\le \tau_{n_s}$, we have
\begin{equation}\label{eq2.18}
\begin{split}
\frac{d^-}{dt}V_s(t)&=
\sum_{j=1}^dy_{s,j}(t)\dot{y}_{s,j}(t)={\chi_s^+}\sum_{j=1}^dy_{s,j}^2(t)
+\sum_{j=1}^d{y_{s,j}(t)}{\bar{f}_{s,j}(t,y_s(t))}\\
&\le{2\chi_s^+}V_s(t)+2\varepsilon|\bchi_*^+(\bA_{u})|
V_s(t)\\
&=2\left({\chi_s^+}+\varepsilon|\bchi_*^+(\bA_{u})|\right)V_s(t).
\end{split}
\end{equation}
Thus, for any $s=1,2,\dotsc$ and any $\tau_{n_{s-1}}<t\le \tau_{n_s}$, by
the Gronwall inequality (see \cite[Lemma~2.1.2]{BP}, for example) we
can obtain
\begin{subequations}\label{eq2.19}
\begin{align}
V_s(t)&\leq
V_s(\tau_{n_{s-1}})\exp\left\{2\left(\chi_s^++\varepsilon|\bchi_*^+(\bA_{u})|\right)(t-\tau_{n_{s-1}})\right\}.\label{eq2.19a}\\
\intertext{Particularly, at the time instant $t=\tau_{n_{s}}$ for $s=1,2,\dotsc$,
we have} V_s(\tau_{n_{s}})&\leq
V_s(\tau_{n_{s-1}})\exp\left\{2(\chi_s^++\varepsilon|\bchi_*^+(\bA_{u})|)(\tau_{n_{s}}-\tau_{n_{s-1}})\right\}.\label{eq2.19b}
\end{align}
\end{subequations}
Repeatedly applying (\ref{eq2.19b}) yields
\begin{equation}\label{eq2.20}
V_s(\tau_{n_s})\le
V_1(\tau_{n_0})\exp\left\{\sum_{\ell=1}^{s}2(\chi_\ell^++\varepsilon|\bchi_*^+(\bA_{u})|)
(\tau_{n_{\ell}}-\tau_{n_{\ell-1}})\right\}.
\end{equation}
Notice that $y_{1, j}(\tau_{n_0})=x_{0, j}$ for $j=1,\dotsc, d$, where
$x_0=(x_{0, 1}, \dotsc, x_{0, d})^\mathrm{T}\in\mathbb{R}^d$ as the
initial value of the solution $x(t)$ to $(\bA,\bF)_{u}$. Thus
\begin{equation*}
V_1(\tau_{n_0})=\frac12\sum_{j=1}^dy_{1,j}^2(\tau_{n_0})=\frac12\|x_0\|^2
\end{equation*}
and furthermore, for $\tau_{n_{s-1}}<t\le \tau_{n_s}$,
\begin{equation*}
V_s(t)\le
\frac{1}{2}\|x_0\|^2\exp\left\{2\left(\varepsilon|\bchi_*^+(\bA_{u})|t+\chi_s^+(t-\tau_{n_{s-1}})+\sum_{\ell=1}^{s}\chi_\ell^+(\tau_{n_{\ell}}-\tau_{n_{\ell-1}})\right)
\right\}.
\end{equation*}
Also according to (\ref{eq2.3}), we know
$|x_j(t,x_0)|\le|y_{s,j}(t)|\exp(\bgamma\varDelta T_{\!*})$ for
$j=1,\dotsc,d$ and $t>0$. Thus for any $\tau_{n_{s-1}}<t\le \tau_{n_s}$
for $s=1,2,\dotsc$, we have
\begin{equation*}
\begin{split}
\|x(t)\|&\le \|x_0\|\exp\left\{\bgamma\varDelta
T_{\!*}+2[\varepsilon|\bchi_*^+(\bA_{u})|t+\chi_s^+(t-\tau_{n_{s-1}})+\sum_{\ell=1}^s\chi_\ell^+(\tau_{n_{\ell}}-\tau_{n_{\ell-1}})]
\right\}.
\end{split}
\end{equation*}
Therefore, there follows that
\begin{equation}\label{eq2.23}
\blambda(x_0)=\limsup_{t\to+\infty}\frac1t\log\|x(t,x_0)\|\le\varepsilon|{\bchi}_*^+(\bA_{u})|+{\bchi}_*^+(\bA_{u})<0
\end{equation}
for any nonzero vector $x_0\in\Real^d$ and that $x(t,x_0)$ is
asymptotically exponentially stable.

Next, we notice that if $L<\bdelta$ where $\bdelta$ is as in the theorem, then we can always choose some
$\varepsilon$ with $0<\varepsilon<1$ so that
$L\le\bdelta_\varepsilon$.

Now we assume that $A_i(t)$ are upper-triangular for all $i\in\mathcal{I}$, i.e.
\begin{equation*}
A_i(t)=\left[\begin{matrix}a_i^{11}(t)&a_i^{12}(t)&\cdots & a_i^{1d}(t)\\
0 &a_i^{22}(t)&\cdots & a_i^{2d}(t)\\
\vdots & \vdots &\ddots & \vdots\\
0 &0&\cdots & a_i^{dd}(t)
\end{matrix}\right].
\end{equation*}
Denote $D_\gamma=\mathrm{diag}(\gamma, \gamma^2, \cdots, \gamma^d)$ where $\gamma$ is a positive real number. Then notice
\begin{equation*}\begin{split}
D_\gamma^{-1}A_i(t)D_\gamma&=\left[\begin{matrix}a_i^{11}(t)& \gamma a_i^{12}(t)&\cdots & \gamma^{d-1}a_i^{1d}(t)\\
0 &a_i^{22}(t) &\cdots & \gamma^{d-2} a_i^{2d}(t)\\
\vdots & \vdots &\ddots & \vdots\\
0 &0&\cdots & a_i^{dd}(t)
\end{matrix}\right]
=\left[\begin{matrix}a_i^{11}(t)&0 &\cdots & 0\\
0 &a_i^{22}(t)&\cdots & 0\\
\vdots & \vdots &\ddots & \vdots\\
0 &0&\cdots & a_i^{dd}(t)\\
\end{matrix}\right]+\left[\begin{matrix}0& \gamma a_i^{12}(t)&\cdots & \gamma^{d-1}a_i^{1d}(t)\\
0 &0&\cdots & \gamma^{d-2}a_i^{2d}(t)\\
\vdots & \vdots &\ddots & \vdots\\
0 &0&\cdots & 0\\
\end{matrix}\right],
\end{split}\end{equation*}
thus, after a similar transformation, the upper triangular case can be viewed as the previous case with an additional perturbation
\[
\left[\begin{matrix}0& \gamma a_i^{12}(t) &\cdots & \gamma^{d-1}a_i^{1d}(t)\\
0 &0&\cdots & \gamma^{d-2}a_i^{2d}(t)\\
\vdots & \vdots &\ddots & \vdots\\
0 &0&\cdots & 0\\
\end{matrix}\right]
\]
which can be set arbitrarily small when $\gamma$ is small enough. Thus, after the coordinates transformation $y=D_{\gamma}^{-1}x$, for those nonlinear perturbations which satisfy
\[
\|D_\gamma^{-1}F_i(t,D_\gamma y)\|\le \bL(t)\gamma^{1-d}\|y\|\qquad\forall y\in\mathbb{R}^d
\]
the conclusion follows by the argument of the first part if $\bL(t)<\bdelta\gamma^{d-1}$.

Therefore, the proof of Theorem~\ref{thm3.1} is completed.
\end{proof}

Since under the hypothesis of Theorem~\ref{thm3.1} every subsystems are
non-autonomous and not necessarily Lyapunov asymptotically stable,
the standard method of Lyapunov functions for
proving asymptotic stability is hardly applied (cf.~\cite[\S4.1]{BP}). In addition, since here lacks the regularity condition of $\bA_u$, the classical theory of Lyapunov exponents~\cite{BVGN} is invalid for the stability of $(\bA,\bF)_u$ even if $\|F_i(t,x)\|\le\bL\|x\|^{1+\varepsilon}$. Theorem~\ref{thm3.1} shows
the importance of the new Liao-type exponents defined above, since $\blambda(\bA_{u})<0$ cannot guarantee the stability of a system $(\bA,\bF)_u$.

To illustrate Theorem~\ref{thm3.1}, we consider the following simple example.

\begin{Example}\label{example3.2}
Let
\begin{equation*}
B_0(t)=\left[\begin{matrix}1+\gamma_0(t)&\alpha_0(t)\\
\beta_0(t)&-2+\eta_0(t)\end{matrix}\right]\quad \textrm{and}\quad B_1(t)=\left[\begin{matrix}-2+\gamma_1(t)&\alpha_1(t)\\
\beta_1(t)&1+\eta_1(t)\end{matrix}\right]
\end{equation*}
where $\gamma_0(t), \gamma_1(t), \eta_0(t), \eta_1(t),
\alpha_0(t),\alpha_1(t), \beta_0(t),\beta_1(t)$ all are functions
piecewise continuous on $\mathbb{R}_+$ converging to $0$ as
$t\to+\infty$, and let $\sigma$ and $\tau$ be given as in
Example~\ref{example2.2}. Then, we have the following switched system:
\begin{equation}\label{eq2.22}
\dot{x}(t)=B_{\!u_{\sigma,\tau}(t)}(t)x(t),\quad x(0)=x_0\in\mathbb{R}^2\textrm{ and }
t\in\mathbb{R}_+.
\end{equation}
Let us denote
\begin{equation*}
F_0(t,x)=\left[\begin{matrix}\gamma_0(t)&\alpha_0(t)\\
\beta_0(t)&\eta_0(t)\end{matrix}\right]x\quad \textrm{and}\quad F_1(t,x)=\left[\begin{matrix}\gamma_1(t)&\alpha_1(t)\\
\beta_1(t)&\eta_1(t)\end{matrix}\right]x,\quad\forall (t,x)\in\mathbb{R}_+\times\mathbb{R}^2.
\end{equation*}
Then, (\ref{eq2.22}) is equivalent to
\begin{equation*}
\dot{x}(t)=A_{u_{\sigma,\tau}(t)}x(t)+F_{u_{\sigma,\tau}(t)}(t,x(t)),\qquad x(0)=x_0\in\mathbb{R}^2\textrm{ and }
t\in\mathbb{R}_+,\leqno{(\ref{eq2.22})^\prime}
\end{equation*}
where $A_0$ and $A_1$ are the same as in Example~\ref{example2.2}.
According to Theorem~\ref{thm3.1}, the above system is exponentially stable under the switching control $u_{\sigma,\tau}$.
\end{Example}

%%%%%%%%%%%%%%%%%%%%%%%%%%%%%%%%%%%%%%%%%%%%%%%%%%%%%%%%%%%%%%%%%%%%%%%%%%%%%%%%%%%%%
%%%%%%%%%%%%%%%%%%%%%%%%%%%%%%%%%%%%%%%%%%%%%%%%%%%%%%%%%%%%%%%%%%%%%%%%%%%%%%%%%%%%%

\section{Concluding remarks}\label{sec4}%

In this paper, by further generalizing approaches developed in a series of papers
\cite{L73,Dai,DHX}, we obtain a criterion of asymptotic exponential
stability for the quasi-linear continuous-time switching
system based on a new Liao-type exponent for time-varying switched systems; see Theorem~\ref{thm3.1}. This paper can be viewed an important generalization
of \cite{DHX}. Comparing to \cite{DHX}, its significance includes the following:
\begin{enumerate}
\item[(1)] For a deterministic switching system $(\bA,\bF)_{u}$, the
formal ``linear approximations" of its subsystems are time dependent;

\item[(2)]  the new Liao-type exponent $\bchi_*^+(\bA_{u})$  defined in this paper can deal with broader class of systems than the one given in \cite{DHX};
\end{enumerate}

A Liao-type exponent $\bchi_*^+$ of a deterministic system is
generally more useful than its corresponding maximal Lyapunov exponent $\blambda$ for perturbations. More specifically,
$\bchi_*^+<0$  provides us more accurate information than $\blambda<0$. On the other hand, it is worthy to mention that from the viewpoint of ergodic theory, the Liao-type exponents $\bchi_*^+$ can approach
arbitrarily the Lyapunov exponents $\blambda$ according to the subsequent paper \cite{Dai-siam}.

Furthermore, the proposed Liao-type exponent presented in this paper can guide us to
choose a stabilizing switching control for nonlinear switched systems.

%----------------------------------------------------------------
%\enlargethispage{12pt}
%\subsection*{Acknowledgment}
%The authors would like to gratefully thank the anonymous reviewers
%for their very helpful suggestions during the improvement of this
%manuscript.

%%%%%%%%%%%%%%%%%%%%%%%%%%%%%%%%%%%%%%%%%%%%%%%%%%%%%%%%%%%%%%%%%%%%%%%%%%

\end{document}